\documentclass[11pt,letter]{article}
\usepackage{authblk}
\usepackage[margin=1in]{geometry}
\usepackage[utf8x]{inputenc}
\usepackage{amsmath, amsthm}
\usepackage{wrapfig,floatflt,graphicx,amssymb,textcomp,array,amsmath}
\usepackage{arcs}
\usepackage{enumerate}
\usepackage{multirow}
\usepackage{tabularx}
\usepackage{color}
\usepackage{todonotes}
\usepackage[titletoc,title]{appendix}
\newcommand{\etal}{{et~al.}}

\title{Euclidean Bottleneck Bounded-Degree Spanning Tree Ratios
}

\author{Ahmad Biniaz\thanks{Part of this work has been done while the author was an NSERC postdoctoral fellow at University of Waterloo.}
}

\affil{School of Computer Science\\University of Windsor\\ \texttt{ahmad.biniaz@gmail.com}}

\date{}
\newtheorem{lemma}{Lemma}
\newtheorem{corollary}{Corollary}

\newtheorem{theorem}{Theorem}

\newtheorem*{problem*}{Problem}
\newtheorem*{claim*}{Claim}
\newtheorem*{invariant*}{Invariant}

\begin{document}
	\maketitle
	\begin{abstract}
	Inspired by the seminal works of Khuller~\etal~(STOC 1994) and Chan (SoCG 2003) we study the bottleneck version of the Euclidean bounded-degree spanning tree problem.
	A bottleneck spanning tree is a spanning tree whose largest edge-length is minimum, and a bottleneck degree-$K$ spanning tree is a degree-$K$ spanning tree whose largest edge-length is minimum. Let $\beta_K$ be the  supremum ratio of the largest edge-length of the bottleneck degree-$K$ spanning tree to the largest edge-length of the bottleneck spanning tree, over all finite point sets in the Euclidean plane. It is known that $\beta_5=1$, and it is easy to verify that $\beta_2\geqslant2$, $\beta_3\geqslant\sqrt{2}$, and $\beta_4>1.175$.
	
	It is implied by the Hamiltonicity of the cube of the bottleneck spanning tree that $\beta_2\leqslant 3$. The degree-3 spanning tree algorithm of Ravi~\etal~(STOC 1993) implies that $\beta_3\leqslant 2$.   
	Andersen and Ras (Networks, 68(4):302–314, 2016) showed that $\beta_4\leqslant \sqrt{3}$. We present the following improved bounds: $\beta_2\geqslant\sqrt{7}$, $\beta_3\leqslant \sqrt{3}$, and $\beta_4\leqslant \sqrt{2}$. As a result, we obtain better approximation algorithms for Euclidean bottleneck degree-3 and degree-4 
	spanning trees. As parts of our proofs of these bounds we present some structural properties of the Euclidean minimum spanning tree which are of independent interest.
	\end{abstract}
\section{Introduction}

The problem of computing a spanning tree of a graph that satisfies given constraints has been well studied. For example, the famous minimum spanning tree (MST) problem asks for a spanning tree with minimum total edge-length, and the bottleneck spanning tree (BST) problem asks for a spanning tree whose largest edge-length is minimum.  
In the last decades, a number of works have been devoted to the study of low-degree spanning trees with short edges. These trees not only satisfy interesting theoretical properties, but also have applications in wireless networks because nodes with high degree or high transmission range lead to a higher level of interference. For finite point sets in any metric space, one can construct a degree-2 spanning tree (a spanning path) whose largest edge-length is at most thrice the largest edge-length of the BST. Such a tree always exists in the cube\footnote{The cube of a graph $G$ has the same vertices as $G$, and has an edge between two distinct vertices if and only if there exists a path, with at most three edges, between them in $G$.} of the BST, and can be computed in polynomial time \cite{Karaganis1968,Lesniak1973}. This yields a factor-3 approximation algorithm for the bottleneck traveling salesman path problem. We will show that if we use BST's largest edge-length as the lower bound, then it is impossible to obtain a ratio better than $\sqrt{7}$ even in the Euclidean metric in the plane.

This paper addresses the bottleneck degree-$K$ spanning tree problem which is a generalization of the bottleneck traveling salesman path problem (for which $K$ is $2$): given a finite point set in the Euclidean plane and an integer $K\geqslant 2$, find a spanning tree of maximum degree at most $K$ that minimizes the largest edge-length, i.e., the length of the longest edge. The degree constraint is natural to consider as high-degree nodes in wireless networks lead to a higher level of interference. 
The edge-length constraint is also natural to consider, since nodes with high transmission range require higher transmission power which in turn increases the interference.

For $K\geqslant 2$, let $\beta_K$ be the supremum ratio of the largest edge-length of the bottleneck degree-$K$ spanning tree (degree-$K$ BST) to the
largest edge-length of the BST, over all finite point sets in the Euclidean plane. Based on the above discussion, $\beta_2\leqslant 3$. 
The definition of $\beta_K$ is consistent with Chan's definition \cite{Chan2004} of $\tau_K$ as the supremum ratio of the weight of the minimum degree-$K$ spanning tree (degree-$K$ MST) to the weight of the MST, over all finite point sets in the Euclidean plane. 
Since every point set has an MST of degree at most 5 \cite{Monma1992}, we get $\tau_K=1$ for all $K\geqslant 5$. Moreover, since every MST is a BST \cite{Camerini1978}, we get $\beta_K=\tau_K$ for all $K\geqslant 5$. For every $K\in\{2,3,4\}$ there are point sets for which no degree-$K$ MST is a degree-$K$ BST, e.g., for the point set in Figure~\ref{MST-BST-fig} any degree-$K$ MST takes the edge of length $1.237$ while no degree-$K$ BST takes that edge. 

\begin{figure}[htb]
	\centering
	\setlength{\tabcolsep}{0in}
	$\begin{tabular}{cccc}
	\multicolumn{1}{m{.25\columnwidth}}{\centering\includegraphics[width=.23\columnwidth]{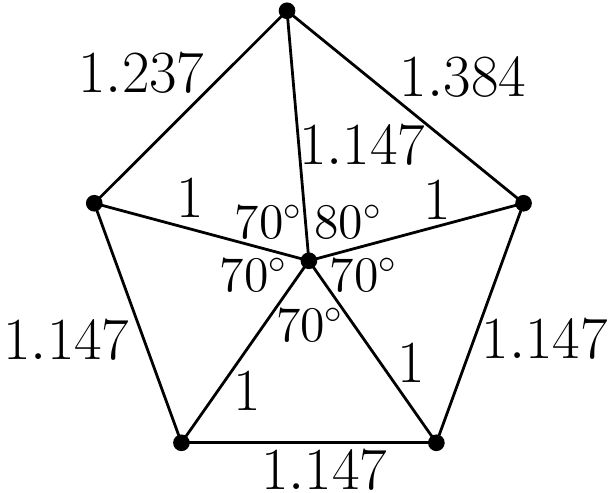}}
	&\multicolumn{1}{m{.25\columnwidth}}{\centering\includegraphics[width=.17\columnwidth]{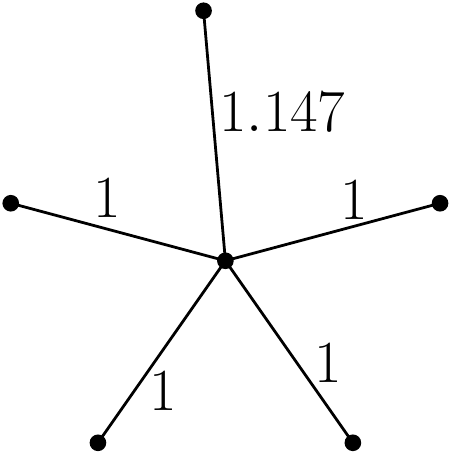}}
	&\multicolumn{1}{m{.25\columnwidth}}{\centering\includegraphics[width=.17\columnwidth]{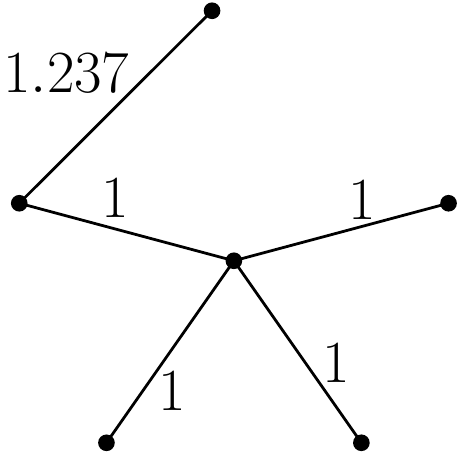}}
	&\multicolumn{1}{m{.25\columnwidth}}{\centering\includegraphics[width=.19\columnwidth]{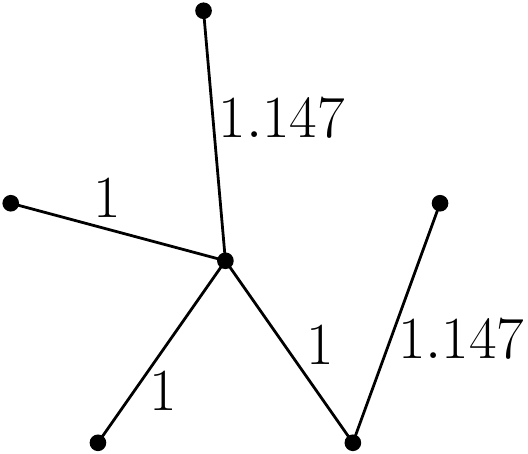}}\\
	input point set&MST and BST&degree-4 MST&degree-4 BST
	\end{tabular}$
	\caption{A point set for which no degree-$K$ MST is a degree-$K$ BST for any $K\in\{2,3,4\}$.}
	\label{MST-BST-fig}
\end{figure}

\subsection{Related work on MST ratios}
The Euclidean degree-$K$ MST problem is NP-hard for $K\in\{2,3,4\}$ \cite{Papadimitriou1977,Papadimitriou1984,Francke2009}.
It is well-known that one can compute a degree-2 spanning tree with weight at most 2 times the MST weight, by doubling the MST edges, computing an Euler tour, and then short-cutting repeated vertices. The constant 2 is tight, as Fekete~\etal~\cite{Fekete1997} showed that for any fixed $\varepsilon >0$ there exist point sets whose degree-2 MST weight is not smaller than $2-\varepsilon$ times the MST weight. Therefore, $\tau_2=2$.

In 1984, Papadimitriou and Vazirani~\cite{Papadimitriou1984} asked whether the Euclidean geometry can be exploited to obtain factors better than 2 for degree-$3$ and degree-$4$ spanning trees. 
Following this question, in 1994, Khuller~\etal~\cite{Khuller1996} achieved the following bounds for $K=3, 4$: $1.103 < \tau_3\leqslant 1.5$ and $1.035 < \tau_4\leqslant 1.25$. The lower bounds are achieved by the center
plus vertices of a square and a regular pentagon, respectively; see Figures~\ref{lower-bounds-fig}(a) and \ref{lower-bounds-fig}(b). The upper bounds are obtained by a recursive algorithm that roots the MST at a leaf $v$, and then transforms it to a degree-$K$ spanning tree with the inductive hypothesis that the root should have degree at most $K-2$ in the new tree. The algorithm transforms the
subtrees rooted at the children of $v$ recursively, and then replaces the star (formed by $v$ and its children) by a small-weight path in which $v$ has degree at most $K-2$. A similar approach was studied before by Ravi~\etal~\cite{Ravi1993} for degree-3 spanning trees in metric spaces.
In 2003, Chan~\cite{Chan2004} revisited the ratios and showed that the upper bounds $1.5$ and $1.25$ are almost tight if we insist the root to have degree at most $K-2$. He managed to improve the upper bounds to $\tau_3<1.402$ and $\tau_4<1.143$ by a weaker inductive hypothesis that the root can have degree at most $K-1$ in the new tree, and by recursing not only on subtrees of the original MST but on trees formed by joining subtrees. A detailed analysis \cite{Jothi2009} shows that Chan's construction of degree-4 spanning trees gives upper bound $\tau_4<1.1381$.
Fekete~\etal~\cite{Fekete1997} conjectured that the lower bounds are tight, that is $\tau_3\approx 1.103$ and $\tau_4\approx 1.035$. 

\subsection{Related work on BST ratios}

Following the result of Itai~\etal~\cite{Itai1982} about the NP-hardness of the Hamiltonian path problem for rectangular grid graphs, Arkin~\etal~\cite{Arkin2009} showed the NP-hardness of this problem for some other classes of grid graphs, including hexagonal grid graphs. This immediately implies the NP-hardness of the Euclidean degree-2 BST problem, and its inapproximability in polynomial time by a factor better than $\sqrt{3}$ unless P=NP. 
Andersen and Ras~~\cite{Andersen2016} have shown that the Euclidean degree-3 BST problem is NP-hard and cannot be approximated in polynomial time by a factor better than $5\sqrt{2}/7$ unless P=NP; they left the status of the corresponding degree-4
problem open.

To obtain lower bounds for $\beta_K$, one could try classic examples (shown in Figure~\ref{lower-bounds-fig}) that achieve lower bounds for similar ratios (e.g., see \cite{Caragiannis2008,Chan2004,Fekete1997,Khuller1996,Parker1984}).
It can be verified simply that $\beta_3\geqslant \sqrt{2}$ and $\beta_4 >1.175$, by the the center plus vertices of a square and a regular pentagon; see Figures~\ref{lower-bounds-fig}(a) and \ref{lower-bounds-fig}(b). The point set in Figure~\ref{lower-bounds-fig}(c) shows that $\beta_2\geqslant 2$; the BST edges have length 1 while any degree-2 BST should pick one of the dashed edges which have length 2.

\begin{figure}[htb]
	\centering
	\setlength{\tabcolsep}{0in}
	$\begin{tabular}{ccc}
	\multicolumn{1}{m{.33\columnwidth}}{\centering\includegraphics[width=.18\columnwidth]{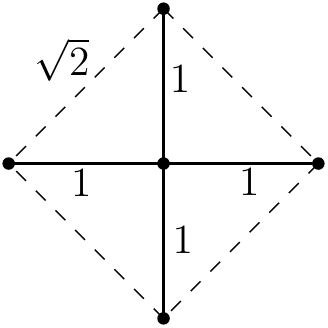}}
	&\multicolumn{1}{m{.33\columnwidth}}{\centering\includegraphics[width=.18\columnwidth]{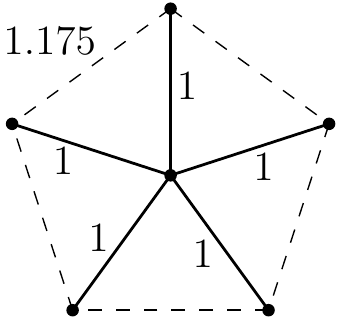}}
	&\multicolumn{1}{m{.33\columnwidth}}{\centering\includegraphics[width=.2\columnwidth]{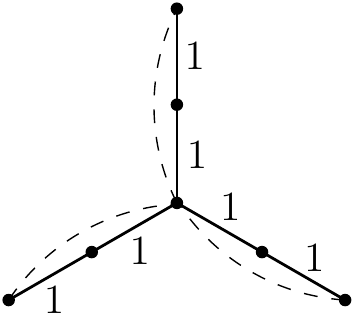}}\\
	(a) degree-3 &(b) degree-4&(c) degree-2
	\end{tabular}$
	\caption{Lower bound examples.}
	\label{lower-bounds-fig}
\end{figure}

The current best known upper bounds are $\beta_2\leqslant 3$, $\beta_3\leqslant 2$, $\beta_4\leqslant \sqrt{3}$ (the bounds $\beta_2 \leqslant 3$ and $\beta_3 \leqslant 2$ hold for any metric space).
As discussed earlier, the upper bound on $\beta_2$ is implied by a known result that the cube of every connected graph (in our case the BST) is Hamiltonian-connected \cite{Karaganis1968,Lesniak1973}; by the triangle inequality the largest edge-length in the cube graph is at most 3 times the largest edge-length in the BST. This is also hinted at in \cite[Exercise 37.2.3]{Cormen1990}. The upper bound on $\beta_3$ is implied by the degree-3 spanning tree algorithm of Ravi~\etal~\cite[Theorem 1.6]{Ravi1993} which replaces the star (formed by the MST's root and its children) with a path; again by the triangle inequality the largest edge-length in the path is at most 2 times the largest edge-length in the star. 
Andersen and Ras~\cite{Andersen2016,Andersen2018a,Andersen2108b} studied bottleneck bounded-degree spanning tree problems from theoretical and experimental points of views. In \cite{Andersen2016}, they obtained similar upper bounds for $\beta_2$ and $\beta_3$, and managed to show that $\beta_4\leqslant \sqrt{3}$ (they obtain this bound by a modified version of Chan's degree-4 spanning tree algorithm).

\subsection{Our contributions}

We focus on ratios of Euclidean bottleneck degree-$K$ spanning trees for $K\in\{2,3,4\}$.
We report the following improved bounds: $\beta_2\geqslant\sqrt{7}$, $\beta_3\leqslant \sqrt{3}$, and $\beta_4\leqslant \sqrt{2}$. For the lower bound, in Section~\ref{degree-2-section} we exhibit a point set in the plane for which the largest edge-length of any degree-2 BST is at least $\sqrt{7}$ times the largest edge-length of the BST. To achieve the upper bounds, we show that for any set of points in the plane there exist
degree-3 and degree-4 spanning trees with edge-lengths within factors $\sqrt{3}$ and $\sqrt{2}$,
respectively, of the BST's largest edge-length. Given the BST (which can be constructed in $O(n\log n)$ time for $n$ points), such trees can be computed in linear time. As a result, we obtain factor-$\sqrt{3}$ and factor-$\sqrt{2}$ approximation algorithms for the Euclidean bottleneck degree-3 and degree-4 spanning tree problems in the plane. The new algorithms are presented in Sections~\ref{degree-4-section} and \ref{degree-3-section}, and a preliminary of their analysis is given in Section~\ref{preliminary-section}. As part of our proofs of these ratios we show some structural properties of the MST which are of independent interest.

The new algorithms are recursive. Even though they are not complicated, the analysis of the degree-3 algorithm is rather involved. In contrast to previous bounded-degree spanning tree algorithms that recurse only on subtrees or joint-subtrees rooted at the children of the root, our algorithms recurse on these subtrees together with the edges connecting them to the root. Moreover, our degree-3 algorithm recurses not only on the children of the root, but also on its grandchildren.
The new algorithms also exploit the geometry of the Euclidean plane, and do not rely only on short-cutting and the triangle inequality which are the main hammers used for the known upper bounds $\beta_2\leqslant 3$ and $\beta_3\leqslant 2$.    

%To add to the importance of the study the largest edge-length constraint, we refer an interested reader to the result of Penrose~\cite{Penrose1997} on the largest edge-length of the BST in random geometric graphs. 

To add to the importance of our ratios, we refer to the works of Dobrev~\etal~\cite{Dobrev2012a,Dobrev2012b} and Caragiannis~\etal~\cite{Caragiannis2008} who studied similar ratios for Euclidean bottleneck strongly connected directed graphs of out-degree at most $K$. These works are motivated by the problem of replacing every omnidirectional antenna in a sensor network, with $K$ directional antennae of low transmission range, so that the resulting network is strongly connected.

\section{Preliminaries for the proofs: some MST properties}
\label{preliminary-section}
To facilitate our analysis, in this section we extract a set of structural properties of the minimum spanning tree in the Euclidean plane. To avoid use of fractional radians, we measure angles in degrees. We will frequently use the well-known fact that the angle between any two adjacent MST edges is at least $60^\circ$ (see \cite{Monma1992}) without mentioning it in all places.

For two points $p$ and $q$ in the plane we denote by $pq$ the straight line segment between $p$ and $q$, and by $|pq|$ the Euclidean distance between $p$ and $q$. Consider a vertex $v$ of degree at least 3 in the MST and assume that its incident edges are sorted radially. An {\em angle at $v$} is the angle between two consecutive edges. Two angles are {\em adjacent} if they share a boundary edge, and {\em nonadjacent} otherwise. We denote the degree of $v$ by $\deg(v)$. We say that two MST edges are {\em adjacent} only if they are incident to the same point (regardless of their relative positions in the radial sorting). The following two lemmas (though very simple) turn out to be crucial for our analysis.

\begin{lemma}
	\label{angle-lemma}
	Let $v$ be any vertex of the MST. The following statements hold for the angles at $v$:
	\begin{enumerate}[$(\!$i$)$]
		\item If $\deg(v)=3$ then there exists an angle that is at most $120^\circ$.
		\item If $\deg(v)=4$ then there exist two nonadjacent angles that are at most $90^\circ$ and $120^\circ$.
		\item If $\deg(v)=5$ then all angles are at most $120^\circ$, and there exist two nonadjacent angles that are at most $90^\circ$.  
	\end{enumerate} 
\end{lemma}	
\begin{proof}
	If $\deg(v)=3$ then the smallest angle at $v$ is at most $120^\circ$.
	
	Assume that $\deg(v)=4$. Let $\alpha_1,\alpha_2,\alpha_3$, and $\alpha_4$ be the angles at $v$ that are ordered radially. Without loss of generality assume that $\alpha_1+\alpha_3\leqslant \alpha_2+\alpha_4$, and thus $\alpha_1+\alpha_3\leqslant 180^\circ$. We claim that $\alpha_1$ and $\alpha_3$, which are nonadjacent, satisfy the angle constraints. Without loss of generality assume that $\alpha_1\leqslant \alpha_3$, and thus $\alpha_1\leqslant 90^\circ$. We already know that $\alpha_1\geqslant 60^\circ$. Thus $\alpha_3\leqslant 120^\circ$. 
	
	Assume that $\deg(v)=5$. The first part of the statement is implied by the fact that every angle at $v$ is at least $60^\circ$. For the second part, let $\alpha$ be the smallest angle at $v$, and observe that $60^\circ\leqslant \alpha\leqslant 72^\circ$. The sum of $\alpha$ and its two adjacent angles is at least $180^\circ$. Thus, $\alpha$ and its smallest nonadjacent angle (which is at most $90^\circ$) are desired angles. 
\end{proof}

\begin{lemma}
	\label{length-lemma}
	Let $pu$ and $uv$ be two adjacent MST edges and let $\alpha$ denote the convex angle between them. Then  $|pv|\leqslant 2\sin(\alpha/2)\cdot\max\{|pu|,|uv|\}$.
\end{lemma}
\begin{proof}
	Observe that $\alpha\geqslant 60^\circ$. After a suitable relabeling assume that $|pu|\leqslant|uv|$. Consider the ray emanating from $u$ and passing through $p$. Let $p'$ be the point on this ray such that $|p'u|=|uv|$; see Figure~\ref{Angelini-fig}(a). Then, the triangle $\bigtriangleup up'v$ is isosceles, and thus $|p'v|=2\sin(\alpha/2)|uv|$. The right hand side is at least $|uv|$ because $\alpha\geqslant 60^\circ$. Therefore, the diameter of $\bigtriangleup up'v$ is $|p'v|$. Since $pv$ lies in this triangle, its length is not more than the diameter. Thus, $|pv|\leqslant |p'v|=2\sin(\alpha/2)|uv|$.
\end{proof}

\begin{figure}[htb]
\centering
\setlength{\tabcolsep}{0in}
$\begin{tabular}{cc}
\multicolumn{1}{m{.5\columnwidth}}{\centering\includegraphics[width=.16\columnwidth]{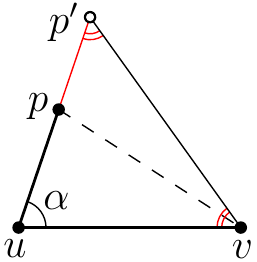}}
&\multicolumn{1}{m{.5\columnwidth}}{\centering\includegraphics[width=.18\columnwidth]{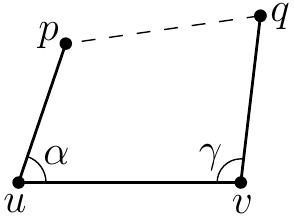}}\\
(a) & (b)
\end{tabular}$
	\caption{Illustrations of (a) Lemma~\ref{length-lemma} and (b) Theorems~\ref{Angelini-thr} and \ref{two-angle-thr}.}
	\label{Angelini-fig}
\end{figure}

Lemma~\ref{angle-lemma} and Lemma~\ref{length-lemma} suffice for the analysis of the degree-4 algorithm. The analysis of the degree-3 algorithm requires stronger tools. Corollary~\ref{Angelini-cor} and Theorem~\ref{two-angle-thr} play important roles in the analysis. Corollary~\ref{Angelini-cor} is implied by the following technical result of Angelini~\etal~\cite{Angelini2014}; see Figure~\ref{Angelini-fig}(b) for an illustration. 

\begin{theorem}[Angelini~\etal~\cite{Angelini2014}]
	\label{Angelini-thr}
	Let $pu$, $uv$, and $vq$ be three MST edges such that both $p$ and $q$ lie on the same side of the line through $uv$. Let $\alpha$ and $\gamma$ denote the convex angles at $u$ and $v$. If $\alpha\leqslant 80^\circ$, then $\gamma\geqslant 120^\circ - \alpha/2$.
\end{theorem}

\begin{corollary}
	\label{Angelini-cor}
	Let $pu$, $uv$, and $vq$ be three MST edges such that both $p$ and $q$ lie on the same side of the line through $uv$. Let $\alpha$ and $\gamma$ denote the convex angles at $u$ and $v$. Then $\alpha+\gamma\geqslant 150^\circ$.
\end{corollary}
\begin{proof}
If both angles are larger than $80^\circ$ then the statement follows immediately. If one of them, say $\alpha$, is at most $80^\circ$ then by Theorem~\ref{Angelini-thr} and the fact that $\alpha,\gamma\geqslant 60^\circ$ we have $\alpha+\gamma\geqslant \alpha+120^\circ-\alpha/2=120^\circ+\alpha/2\geqslant 150^\circ$.
\end{proof} 

We prove Theorem~\ref{two-angle-thr} in Section~\ref{proof-section}. This theorem, which is illustrated in Figure~\ref{Angelini-fig}(b), deals with a maximization problem which has five variables at first glance. We use a sequence of geometric transformations to reduce the number of variables and simplify the proof.

\begin{theorem}
	\label{two-angle-thr}
	Let $pu$, $uv$, and $vq$ be three MST edges such that both $p$ and $q$ lie on the same side of the line through $uv$. Let $\alpha$ and $\gamma$ denote the convex angles at $u$ and $v$. If $\alpha+\gamma\leqslant 210^\circ$, then $|pq|\leqslant \sqrt{3}\cdot\max\{|pu|,|uv|,|vq|\}$.
\end{theorem}

\paragraph{Remark:} The upper bound $210^\circ$ on $\alpha+\gamma$ in the statement of Theorem~\ref{two-angle-thr} is tight in the sense that if we replace it by $(210+\epsilon)^\circ$, then there exist MST edges $pu, uv, vq$ for which the objective inequality does not hold. For example take $\alpha=90^\circ$, $\gamma=(120+\epsilon)^\circ$, and $|uv|=|vq|=1$. Then $|uq|>\sqrt{3}$. By placing $p$ very close to $u$ such that $|pu|$ tends to zero, we make $|pq|$ larger than $\sqrt{3}$.

\section{Degree-4 spanning tree algorithm}
\label{degree-4-section}

Our algorithm is recursive. The algorithm recurses not only on rooted subtrees of the original tree, but on rooted subtrees together with the edges connecting them to their parents.
For a rooted tree $T$ and a single vertex $r\notin T$, we denote by $T+r$ the rooted tree obtained by making the root of $T$ a child of $r$; see Figure~\ref{degree4-fig}(left).

\begin{figure}[htb]
	\centering
	\includegraphics[width=.78\columnwidth]{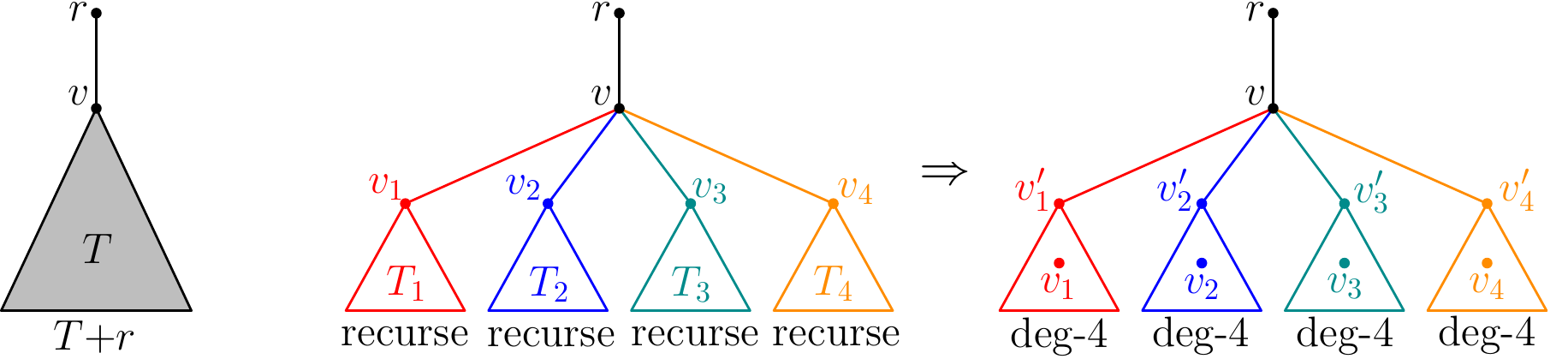}
	\caption{Left: Illustration of the tree $T+r$. Right: Transforming the trees $T_i+v$ recursively.}
	\label{degree4-fig}
\end{figure}

We are given a minimum spanning tree $\mathcal{T}$ of a set of points in the plane, which we may assume \cite{Monma1992} has maximum degree at most $5$. Notice that $\mathcal{T}$ is also a bottleneck spanning tree \cite{Camerini1978}. Root $\mathcal{T}$ at a fixed leaf $r$ so that each vertex has at most four children. Let $v$ denote the only child of $r$ and let $T$ denote the subtree rooted at $v$. Then $\mathcal{T}=T+r$, as in Figure~\ref{degree4-fig}(left). Let $b(T+r)$ denote the largest edge-length of $T+r$. 
Our recursive algorithm transforms the rooted tree $T+r$ into a new degree-4 spanning tree, with the inductive hypothesis that
\begin{center}
\begin{minipage}[c][][b]{0.8\textwidth}
	{\em the root $r$ has degree 1 and $v$ has degree at most 3 in the new tree, and the largest edge-length of the new tree is at most $\sqrt{2}\cdot b(T+r)$.}
\end{minipage} 
\end{center}

We note that after transformation, $v$ may not be the child of $r$ in the new tree. 

The algorithm works as follows. After a suitable rescaling we may assume that $b(T+r)=1$. Let $v_1,\dots,v_k$ be the $k (\leqslant 4)$ children of $v$ in $T$ that are ordered radially. Let $T_1,\dots,T_k$ be the subtrees rooted at $v_1,\dots,v_k$. Transform $T_1+v,\dots,T_k+v$ recursively, and let $T'_1+v,\dots,T'_k+v$ be the resulting new degree-4 trees. See Figure~\ref{degree4-fig} for an illustration.  
By the inductive hypothesis, in each $T'_i+v$, the vertex $v$ has degree 1 and the vertex $v_i$ has degree at most 3. 
Let $v'_i$ be the only child of $v$ in each tree $T'_i+v$, and again notice that $v'_i$ might be different from $v_i$. If the child $v'_i$ is different from $v_i$ then we say that $v'_i$ is {\em adopted}, otherwise $v'_i$ is called {\em natural}.

\begin{figure}[htb]
	\centering
	\setlength{\tabcolsep}{0in}
	$\begin{tabular}{cc}
	\multicolumn{1}{m{.5\columnwidth}}{\centering\includegraphics[width=.23\columnwidth]{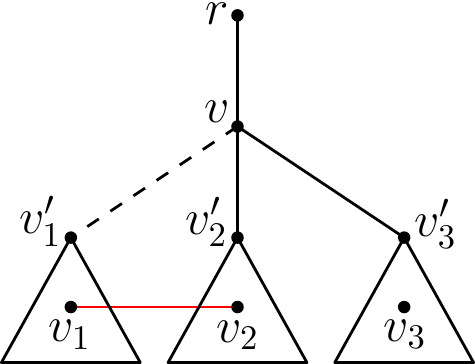}}
	&\multicolumn{1}{m{.5\columnwidth}}{\centering\includegraphics[width=.33\columnwidth]{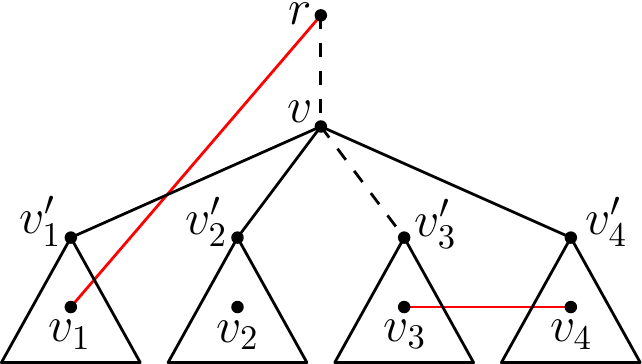}}\\
	(a)&(b)
	\end{tabular}$
	\caption{Local replacement of edges: (a) $k=3$ and $\alpha$ is defined by $vv_1$ and $vv_2$. (b) $k=4$, $\alpha_1$ is defined by $vr$ and $vv_1$, and $\alpha_2$ is defined by $vv_3$ and $vv_4$.}
	\label{deg4-fig}
\end{figure}

The above transformation of trees $T_i+v$ does not change the degrees of $v$ and $r$. Moreover, for every $i\in\{1,\dots,k\}$, we have that $|vv'_i|\geqslant |vv_i|$ because otherwise $vv_i$ should not be an edge of the original minimum spanning tree $\mathcal{T}$.  
After transforming trees $T_i+v$, we replace the edges $vr$, $vv'_1,\dots,vv'_k$ locally to obtain a transformation of $T+r$. To do so, we differentiate between different values of $k$.  
\begin{itemize}
	\item $k\leqslant 2$. In this case $\deg(v)\leqslant 3$ and $\deg(r)=1$. We just leave the edges $rv,vv'_1,vv'_2$ in.
	\item $k=3$. We describe this case in detail. In this case $\deg(v)=4$ and there exists an angle $\alpha\leqslant 90^\circ$ at $v$ in the original tree $T+r$. If $\alpha$ is defined by two edges $vv_i$ and $vv_{i+1}$ then we add the edge $v_iv_{i+1}$ and remove $vv'_i$, as in Figure~\ref{deg4-fig}(a). After this replacement, $v$ has degree $3$, $r$ has degree $1$, and each of $v_i$ and $v_{i+1}$ has degree at most $4$. Moreover, by Lemma~\ref{length-lemma} the length of the new edge $v_iv_{i+1}$ is at most $2\sin(\alpha/2)\cdot\max\{|vv_i|,|vv_{i+1}|\}\leqslant 2\sin(\alpha/2)\leqslant \sqrt{2}$.
	
	If $\alpha$ is defined by $vr$ and an edge $vv_i$ then we add $rv_i$ and remove $rv$. After this replacement, $v$ has degree $3$, $r$ has degree $1$, and $v_i$ has degree at most $4$. Again by Lemma~\ref{length-lemma} we have $|rv_i|\leqslant 2\sin(\alpha/2)\cdot\max\{|vr|,|vv_i|\}\leqslant \sqrt{2}$.
	\item $k=4$. In this case $\deg(v)=5$, and thus by Lemma~\ref{angle-lemma} there exist two nonadjacent angles $\alpha_1, \alpha_2\leqslant 90^\circ$ at $v$ in the original tree $T+r$. We process $\alpha_1$ as follows: If $\alpha_1$ is defined by two edges $vv_i$ and $vv_{i+1}$ then add $v_iv_{i+1}$ and remove $vv'_i$, but if $\alpha_1$ is defined by $vr$ and an edge $vv_i$ then add $rv_i$ and remove $rv$. We process $\alpha_2$ analogously. See Figure~\ref{deg4-fig}(b). After processing both angles, $v$ has degree $3$, $r$ has degree $1$, and each $v_i$ has degree at most $4$. It is implied by Lemma~\ref{length-lemma} that the length of each new edge is at most $\sqrt{2}$.
\end{itemize}
Therefore, we obtain a new tree that satisfies the inductive hypothesis, and thus a ratio of $\sqrt{2}$ has been established. The above local replacements take constant time per root. Thus, given the initial degree-5 MST (which is also a BST and can be constructed in $O(n \log n)$ time for $n$ points \cite{Monma1992}), the algorithm runs in linear time.

\paragraph{Remark:} Our analysis of the ratio $\sqrt{2}$ is tight under our inductive hypothesis that ``the root $r$ must have degree $1$ and $v$ must have degree at most $3$ in the new tree''; the example in Figure~\ref{lower-bounds-fig}(a) indicates why.

\section{Degree-3 spanning tree algorithm}
\label{degree-3-section}

Let $T+r$ be a degree-5 minimum spanning tree that is rooted at a leaf $r$, and let $v$ be the only child of $r$. 
Our approach for degree-3 spanning trees is similar to that of degree-4 trees, except the degree of $v$ that should be at most 2. The algorithm transforms $T + r$ into a new degree-3 spanning tree, with the inductive hypothesis that
\begin{center}
	\begin{minipage}[c][][b]{0.8\textwidth}
		{\em the root $r$ has degree 1 and $v$ has degree at most 2 in the new tree, and the largest edge-length of the new tree is at most $\sqrt{3}\cdot b(T+r)$.}
	\end{minipage} 
\end{center}

Assume that $b(T+r)=1$. Let $v_1,\dots,v_k$ be the children of $v$ that are ordered radially, and let $T_{v_1},\dots,T_{v_k}$ be the subtrees rooted at these vertices, respectively (in this section the name $T_{v_i}$ is more convenient than $T_i$). Logically, similar to the degree-4 algorithm, our degree-3 algorithm should first transform the trees $T_{v_i}+v$ and then replace the edges incident to $v$ locally. However, when $k=4$ (i.e., when $\deg(v)=5$) we are not able to replace three of the incident edges without breaking the degree constraint for $r$ or for a $v_i$. As such we differentiate between two cases where $k\leqslant 3$ and $k=4$. 

\vspace{8pt}
\noindent{\bf Case $k\leqslant 3$.} Transform $T_{v_1}+v,\dots,T_{v_k}+v$ recursively to obtain new degree-3 trees. Let $v'_1,\dots,v'_k$ be the children of $v$ in the new trees. After this transformation, the degree of each $v_i$ is at most 2, and the degrees of $r$ and $v$ remain unchanged. Moreover, for each $i$ we have $|vv'_i|\geqslant |vv_i|$. To obtain a transformation of $T+r$, we then replace $vr,vv'_1,\dots,vv'_k$ locally. 

\begin{figure}[htb]
	\centering
	\setlength{\tabcolsep}{0in}
	$\begin{tabular}{cc}
	\multicolumn{1}{m{.5\columnwidth}}{\centering\includegraphics[width=.15\columnwidth]{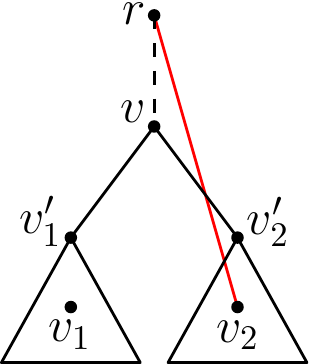}}
	&\multicolumn{1}{m{.5\columnwidth}}{\centering\includegraphics[width=.25\columnwidth]{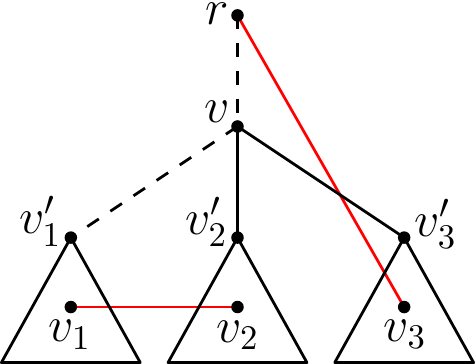}}\\
	(a)&(b)
	\end{tabular}$
	\caption{Local replacement of edges: (a) $k=2$ and $\alpha$ is defined by $vr$ and $vv_2$. (b) $k=3$, $\alpha_1$ is defined by $vv_1$ and $vv_2$, and $\alpha_2$ is defined by $vr$ and $vv_3$.}
	\label{deg3-fig}
\end{figure}

\begin{itemize}
	\item $k\leqslant 1$. In this case $\deg(v)\leqslant 2$ and $\deg(r)=1$. We just leave the edges $rv,vv'_1$ in.
	\item $k=2$. Then $\deg(v)=3$. By Lemma~\ref{angle-lemma} there exists an angle $\alpha\leqslant 120^\circ$ at $v$ in the original tree $T+r$. If $\alpha$ is defined by two edges $vv_i$ and $vv_{i+1}$ then add $v_iv_{i+1}$ and remove $vv'_i$. If $\alpha$ is defined by $vr$ and an edge $vv_i$ then add $rv_i$ and remove $rv$, as in Figure~\ref{deg3-fig}(a). In either case, after the replacement, $v$ has degree $2$, $r$ has degree $1$, and each $v_i$ has degree at most $3$. Moreover, by Lemma~\ref{length-lemma} the length of the new edge is at most $2\sin(\alpha/2)\leqslant \sqrt{3}$.
	\item $k=3$. Then $\deg(v)=4$. By Lemma~\ref{angle-lemma} there exist two nonadjacent angles $\alpha_1, \alpha_2\leqslant 120^\circ$ at $v$ in $T+r$. We process each of $\alpha_1$ and $\alpha_2$ similar to $\alpha$ in the previous case; see Figure~\ref{deg3-fig}(b). After processing both angles, $v$ has degree $2$, $r$ has degree $1$, each $v_i$ has degree at most $3$, and by Lemma~\ref{length-lemma} the lengths of new edges are at most $\sqrt{3}$.
\end{itemize}

%\vspace{8pt}
\noindent
{\bf Case $k=4$.} Here is the place where we need more technical results. To see the difficulty of this case we refer to the importance of the non-adjacency of $\alpha_1$ and $\alpha_2$ in case $k=3$. Since these two angles are nonadjacent we were able to replace two incident edges (to $v$) without increasing the degree of each $v_i$ by more than $1$. In the current case, $\deg(v)=5$, and thus to satisfy the degree constraint for $v$ we need to replace three incident edges. However, there are only five angles at $v$, and thus we are unable to find three nonadjacent angles. 

It might be tempting to attach two new edges to a vertex $v_i$ and remove the edge $vv_i$; this would increase the degree of $v_i$ by at most $1$. We should be careful here as the edge $vv_i$ may not be present after transforming the trees $T_{v_i}+v$ because all children of $v$ could be adopted. In this case, we may not even be able to attach adopted children together or to natural children without breaking the edge-length constraint. 

To handle this case, our idea is to recurse not only on the children of $v$, but also on its grandchildren. This gives rise to somewhat lengthier analysis. Also, a technical complication arises because now we need to bound the distance between endpoints of two nonadjacent MST edges.

Let $v_1,\dots,v_4$ be the children of $v$ in counterclockwise order around $v$, such that their grandparent $r$ lies between $v_1$ and $v_4$. Let $\alpha_1=\angle rvv_1$, $\alpha_2=\angle v_1vv_2$, $\alpha_3=\angle v_2vv_3$, $\alpha_4=\angle v_3vv_4$, and $\alpha_5=\angle v_4vr$, as in Figure~\ref{deg3-k5-fig}(a).
Since $\alpha_3\geqslant 60^\circ$, the smallest of $\alpha_1+\alpha_2$ and $\alpha_4+\alpha_5$ is at most $150^\circ$. After a suitable reflection and relabeling assume that $\alpha_1+\alpha_2\leqslant 150^\circ$. Now we are going to recurse on the children of both $v$ and $v_1$. Let $u_1,\dots, u_l$ be the $l(\leqslant 4)$ children of $v_1$ in clockwise order around $v_1$, such that their grandparent $v$ lies between $u_1$ and $u_2$, as in Figure~\ref{deg3-k5-fig}(c).

\begin{figure}[htb]
	\centering
	\setlength{\tabcolsep}{0in}
	$\begin{tabular}{cc}
	\multicolumn{1}{m{.5\columnwidth}}{\centering\includegraphics[width=.33\columnwidth]{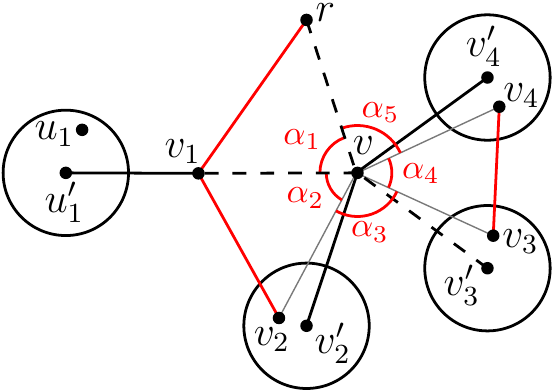}}
	&\multicolumn{1}{m{.5\columnwidth}}{\centering\includegraphics[width=.33\columnwidth]{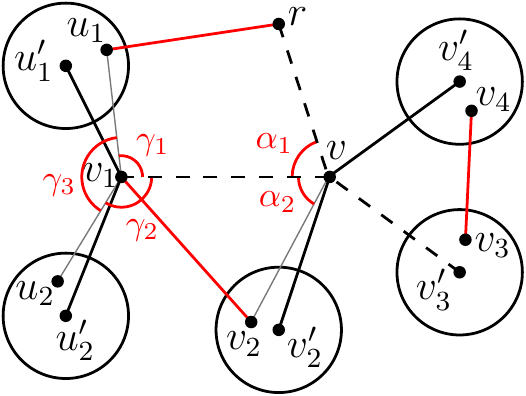}}\\
	(a) \text{$l\leqslant 1$} & (b) \text{$l= 2$}
	\\[.3cm]
	\multicolumn{1}{m{.5\columnwidth}}{\centering\includegraphics[width=.38\columnwidth]{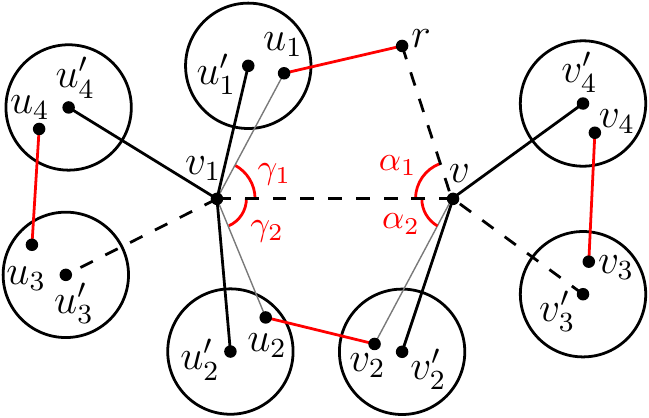}}
	&\multicolumn{1}{m{.5\columnwidth}}{\centering\includegraphics[width=.38\columnwidth]{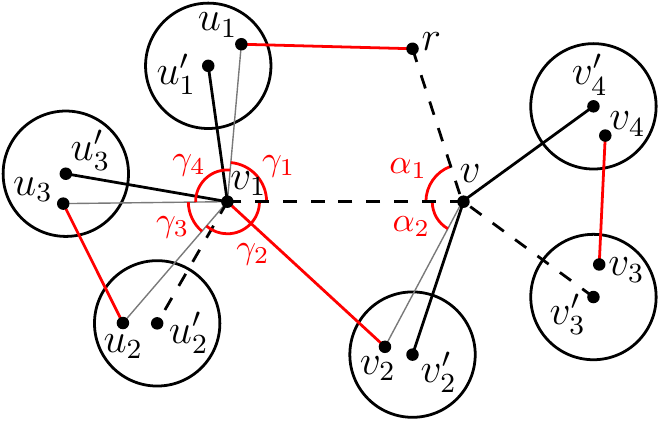}}\\
	(c) \text{$l= 4$} & (d) \text{$l= 3$}
	\end{tabular}$
	\caption{Obtaining degree-3 spanning tree when $k=4$, i.e., when $\deg(v)=5$.}
	\label{deg3-k5-fig}
\end{figure}

Transform $T_{v_2}+v,T_{v_3}+v,T_{v_4} + v$ recursively to obtain new degree-3 trees, and let $v'_2,v'_3, v'_4$ be the children of $v$ in the new trees.
Also transform $T_{u_1}+v_1,\dots,T_{u_l} + v_1$ recursively to obtain new degree-3 trees, and let $u'_1,\dots, u'_l$ be the children of $v_1$ in the new trees. In Figure~\ref{deg3-k5-fig} every new tree is shown by a circle and a connecting black edge to the parent. 
After these transformations, each of $v_2$, $v_3$, $v_4$, $u_1,\dots,u_l$ has degree at most 2, by the inductive hypothesis.
Now we are going to replace the edges incident to $v$ and $v_1$ locally to obtain a transformation of $T+r$. Depending on the value of $l$ we consider four cases. We perform the replacement in such way that at the end of each case the following constraints hold: $\deg(r)=1$, $\deg(v)=2$, $\deg(v_i)\leqslant 3$ for every $i\in\{1,\dots, 4\}$, $\deg(u_j)\leqslant 3$ for every $j\in\{1,\dots,l\}$, and the length of every new edge is at most $\sqrt{3}$. 

Before proceeding to the cases we note that since $\deg(v)=5$, by Lemma~\ref{angle-lemma} all angles at $v$ in $T+r$ are at most $120^\circ$. Thus, by Lemma~\ref{length-lemma} the distance between any two consecutive neighbors of $v$ in $T+r$ is at most $2\sin(60^\circ)=\sqrt{3}$.

\begin{itemize}
	\item $l\leqslant 1$. In this case add $rv_1$, $v_1v_2$, $v_3v_4$ and remove $rv$, $v_1v$, $v'_3v$; see Figure~\ref{deg3-k5-fig}(a). 
		
	\item $l = 2$. Let $\gamma_1=\angle u_1v_1v$, $\gamma_2=\angle vv_1u_2$, and $\gamma_3=\angle u_2v_1u_1$, as in Figure~\ref{deg3-k5-fig}(b). If $\gamma_3\leqslant 120^\circ$, then add $u_1u_2$, $rv_1$, $v_1v_2$, $v_3v_4$ and remove $v_1u'_1$, $rv$, $v_1v$, $v'_3v$. Assume that $\gamma_3\geqslant 120^\circ$. Then $\gamma_1+\gamma_2\leqslant 240^\circ$ and consequently $\alpha_1+\alpha_2+\gamma_1+\gamma_2\leqslant 390^\circ$. Thus we have $\alpha_1 + \gamma_1\leqslant 195^\circ$ or $\alpha_2 + \gamma_2\leqslant 195^\circ$. If $\alpha_1 + \gamma_1\leqslant 195^\circ$ then add $ru_1$, $v_1v_2$, $v_3v_4$ and remove $rv$, $v_1v$, $v'_3v$ (this case is depicted in Figure~\ref{deg3-k5-fig}(b)); notice that $|ru_1|\leqslant \sqrt{3}$ by Theorem~\ref{two-angle-thr}. If $\alpha_2 + \gamma_2\leqslant 195^\circ$ then add $rv_1$, $u_2v_2$, $v_3v_4$ and remove $rv$, $v_1v$, $v'_3v$; notice that $|u_2v_2|\leqslant \sqrt{3}$ by Theorem~\ref{two-angle-thr}.

	\item $l = 4$. We describe this case first, as it is simpler. Since $\deg(v_1) = 5$, it is implied by Lemma~\ref{angle-lemma} and Lemma~\ref{length-lemma} that the distance between any two consecutive neighbors of $v_1$ in $T + r$ is at most $\sqrt{3}$. In this case we add $ru_1$, $u_2v_2$, $u_3u_4$, $v_3v_4$ and remove $rv$, $v_1v$, $u'_3v_1$, $v'_3v$, as in Figure~\ref{deg3-k5-fig}(c). We only need to show that $|ru_1|$ and $|u_2v_2|$ are at most $\sqrt{3}$.
		
	Let $\gamma_1=\angle u_1v_1v$ and $\gamma_2=\angle vv_1u_2$, as in Figure~\ref{deg3-k5-fig}(c). Observe that  $\gamma_1+\gamma_2\leqslant 180^\circ$, and thus $\alpha_1+\alpha_2+\gamma_1+\gamma_2\leqslant 330^\circ$. By Corollary~\ref{Angelini-cor} we have $\alpha_1+\gamma_1\geqslant 150^\circ$ and $\alpha_2+\gamma_2\geqslant 150^\circ$. Combining these inequalities, we ge that $\alpha_1+\gamma_1$ and $\alpha_2+\gamma_2$ are at most $330^\circ - 150^\circ= 180^\circ$. Having these constraints, it is implied by Theorem~\ref{two-angle-thr} that $|ru_1|$ and $|u_2v_2|$ are at most $\sqrt{3}$.

	\item $l = 3$. Let $\gamma_1=\angle u_1v_1v$, $\gamma_2=\angle vv_1 u_2$, $\gamma_3=\angle u_2v_1u_3$, and $\gamma_4=\angle u_3v_1u_1$, as in Figure~\ref{deg3-k5-fig}(d). We differentiate between two cases where (i) $\max\{\gamma_3,\gamma_4\}\geqslant 120^\circ$ and (ii) $\max\{\gamma_3,\gamma_4\}\leqslant 120^\circ$. 
	
	In case (i) we add $ru_1$, $u_2v_2$, $v_3v_4$ and remove $rv$, $v_1v$, $v'_3v$.	
	Observe that $\gamma_1+\gamma_2\leqslant 180^\circ$, and thus as in the previous case ($l=4$) both $\alpha_1+\gamma_1$ and $\alpha_2+\gamma_2$ are at most $180^\circ$. Thus by Theorem~\ref{two-angle-thr} both $|ru_1|$ and $|u_2v_2|$ are at most $\sqrt{3}$. 
	
	Consider case (ii). We already know that $\gamma_3+\gamma_4\geqslant 120^\circ$. By a reasoning similar to the one in case $l=2$ we have  $\alpha_1 + \gamma_1\leqslant 195^\circ$ or $\alpha_2 + \gamma_2\leqslant 195^\circ$. If $\alpha_1 + \gamma_1\leqslant 195^\circ$ then add $ru_1$, $u_2u_3$, $v_1v_2$, $v_3v_4$ and remove $rv$, $u'_2v_1$, $v_1v$, $v'_3v$ (this case is depicted in Figure~\ref{deg3-k5-fig}(d)); notice that $|u_2u_3|\leqslant \sqrt{3}$ by Lemma~\ref{length-lemma} and $|ru_1|\leqslant \sqrt{3}$ by Theorem~\ref{two-angle-thr}. If $\alpha_2 + \gamma_2\leqslant 195^\circ$ then add $rv_1$, $u_1u_3$, $u_2v_2$, $v_3v_4$ and remove $rv$, $u'_1v_1$, $v_1v$, $v'_3v$; notice that $|u_1u_3|\leqslant \sqrt{3}$ by Lemma~\ref{length-lemma} and $|u_2v_2|\leqslant \sqrt{3}$ by Theorem~\ref{two-angle-thr}.
\end{itemize}

Therefore, we obtain a new tree that satisfies the inductive hypothesis, and thus a ratio of $\sqrt{3}$ has been established. The above local replacements take constant time per root. Thus, given the initial degree-5 MST, the algorithm runs in linear time.

\paragraph{Remark:} Our analysis of the ratio $\sqrt{3}$ is tight under our inductive hypothesis that ``the root $r$ must have degree $1$ and $v$ must have degree at most $2$ in the new tree''; a set of four points formed by the center plus the vertices of an equilateral triangle, indicates why.

\section{Worst-case ratio for bottleneck degree-2 spanning trees}
\label{degree-2-section}

In this section we show that $\beta_2\geqslant \sqrt{7}$. Recall that any degree-2 spanning tree is a Hamiltonian path and vice versa. It is already known \cite{Caragiannis2008} that the worst-case ratio of the largest edge-length of the bottleneck Hamiltonian cycle to the largest edge-length of the BST is at least $\sqrt{7}$; see the example in Figure~\ref{lower-bounds-fig}(c). This small example, however, does not give a lower bound better than $2$ for $\beta_2$, i.e., for the bottleneck Hamiltonian path.

\begin{wrapfigure}{r}{3in} 
	\vspace{-12pt} 
	\centering
	\includegraphics[width=.45\columnwidth]{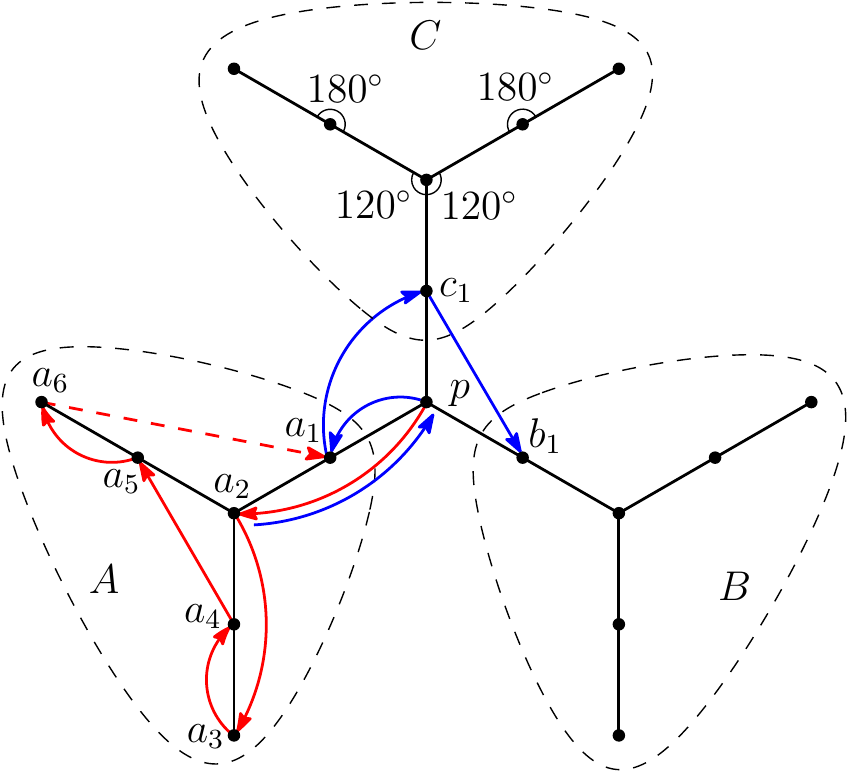}
	%\caption{Illustration of the proof of Theorem~\ref{Delaunay-thr}.}
	%\label{Delaunay-fig}
	\vspace{-12pt} 
\end{wrapfigure}
The figure to the right exhibits a set of 19 points that achieves lower bound $\sqrt{7}$ for the bottleneck Hamiltonian path.  
The figure also shows the MST where every edge has length 1.
Every angle at each degree-3 vertex is $120^\circ$, and every angle at each degree-2 vertex is $180^\circ$. Except the point $p$, all other points are partitioned into three sets $A$, $B$, and $C$.
Consider any bottleneck Hamiltonian path $\delta$ on this point set. We prove that $\delta$ has a ``long edge'', i.e., an edge of length at least $\sqrt{7}$. This would immediately imply that $\beta_2\geqslant \sqrt{7}$. Due to the size of the point set which is fairly large (compared with the lower bound examples in Figure~\ref{lower-bounds-fig}) and our desire to provide a concrete argument, the proof is somewhat lengthy. 

The path $\delta$ has $18$ edges. Let $q$ be an endpoint of $\delta$ such that the number of path edges between $p$ and $q$ is at least $9$. Orient all edges of $\delta$ towards $q$. For any vertex $u\in \delta$, we denote its outgoing edge by $\overrightarrow{u}$.
The edge $\overrightarrow{p}$ goes into one of the three sets, say $A$. For the rest of our argument, we consider two cases depending on whether the incoming edge of $p$ comes from $B\cup C$ or from $A$. 

First assume that the incoming edge of $p$ comes form $B\cup C$.
If $\overrightarrow{p}$ goes to $a_1$ then there is a long edge between $A\setminus\{a_1\}$ and $B\cup C$ (recall the 9 edges following $p$). If $\overrightarrow{p}$ goes into $A\setminus\{a_1,a_2\}$ then $\overrightarrow{p}$ is long. Thus, assume that $\overrightarrow{p}$ goes to $a_2$, as depicted by a red edge in the figure. Now consider $\overrightarrow{a_2}$. If $\overrightarrow{a_2}$ goes into $B\cup C$ then it is long, and if it goes to $a_1$ then there is a long edge between $A\setminus\{a_1,a_2\}$ and $B\cup C$. Thus, assume that it goes to a vertex in $A\setminus\{a_1,a_2\}$; by symmetry assume that this vertex is $a_3$ or $a_4$. If it goes to $a_4$ then there would be a long edge incident to $a_3$. Assume that it goes to $a_3$ as in the figure. If $\overrightarrow{a_3}$ goes to any vertex other than $a_4$ then it is long, thus assume that it goes to $a_4$. Now consider $\overrightarrow{a_4}$. If $\overrightarrow{a_4}$ goes to $B\cup C\cup\{a_6\}$ then it is long, and if it goes to $a_1$ then there would be a long edge between $\{a_5,a_6\}$ and $B\cup C$. Thus, assume that it goes to $a_5$. If $\overrightarrow{a_5}$ goes to any vertex other than $a_6$ then there would be a long edge incident to $a_6$, thus assume that it goes to $a_6$. In this setting, the edge $\overrightarrow{a_6}$ is long. Notice that all edges $\overrightarrow{p}$, $\overrightarrow{a_2}$, $\overrightarrow{a_3}$, $\overrightarrow{a_4}$, $\overrightarrow{a_5}$, and $\overrightarrow{a_6}$ exist because there are at least 9 edges from $p$ to $q$.

Now assume that the incoming edge of $p$ comes form $A$. If this edge comes from $\{a_3,a_4,a_5,a_6\}$ then it is long. If this edge comes from $a_1$ then by an argument similar to the one above, we traverse the edges following $\overrightarrow{p}$ until we get a long edge (recall that $\overrightarrow{p}$ goes into $A$). Thus, assume that the incoming edge of $p$ comes from $a_2$, as depicted by a blue edge in the figure. Consider $\overrightarrow{p}$ again. If it goes to any point of $A$ other than $a_1$ then it is long. Assume that it goes to $a_1$, as in the figure. If $\overrightarrow{a_1}$ goes into the set $A\setminus\{a_1,a_2\}$ then there is a long edge between this set and $B\cup C$. Thus, assume that $\overrightarrow{a_1}$ goes into $B$ or $C$; by symmetry assume $C$. At this point notice that any edge between $A\setminus\{a_1,a_2\}$ and $B\cup C$ is long, and thus we may assume there is no edge between these two sets. If $\overrightarrow{a_1}$ goes to any point of $C$ other than $c_1$ then it is long, and thus assume that it goes to $c_1$. 
If $\overrightarrow{c_1}$ goes into the set $C\setminus\{c_1\}$ then there is a long edge between this set and $B$ (considering the 9 edges following $p$). If it goes into $B\setminus \{b_1\}$ then it is long. Thus, assume that it goes to $b_1$, as in the figure. In this setting there must be an edge between $B$ and $C\setminus\{c_1\}$ (even if the next edges of the path capture all remaining points of $B$, the 9th edge has to leave $B$); any such edge is long. This is the end of the proof.   

\section{Proof of Theorem~\ref{two-angle-thr}}
\label{proof-section}

In this section we prove Theorem~\ref{two-angle-thr} that: {\em Let $pu$, $uv$, and $vq$ be three MST edges such that both $p$ and $q$ lie on the same side of the line through $uv$. Let $\alpha$ and $\gamma$ denote the convex angles at $u$ and $v$. If $\alpha+\gamma\leqslant 210^\circ$, then $|pq|\leqslant \sqrt{3}\cdot\max\{|pu|,|uv|,|vq|\}$.}

This theorem states a maximization problem with five variables $|pu|$, $|uv|$, $|vq|$, $\alpha$, and $\gamma$. We use a sequence of geometric transformations to discretize the problem, reduce the number of variables, and simplify the proof. To do so we use Lemma~\ref{moving-point-lemma} and the following result of Abu-Affash~\etal~\cite{Abu-Affash2015}. 

\begin{lemma}[Abu-Affash~\etal~\cite{Abu-Affash2015}]
	\label{Abu-Affash-lemma}
	If $pu$ and $uv$ are two adjacent MST edges, then the triangle with vertices $p$, $u$, and $v$ has no other vertex of the MST in its interior or on its boundary.
\end{lemma}

\begin{lemma}
	\label{moving-point-lemma}
	Let $p$ and $v$ be two distinct points in the plane and let $R$ be a ray emanating from $v$ that is not passing through $p$. Let $c_1$ and $c_2$ be two constants such that $c_2>c_1>0$. Then the largest value of $|pq|/\max\{|vq|,c_2\}$ over all points $q$ on $R$ with $|pq|>|vq|\geqslant c_1$, is achieved when $|vq|=c_1$ or $|vq|=c_2$. 
\end{lemma}

\begin{wrapfigure}{r}{1.5in} 
	\centering
	\vspace{2pt} 
	\includegraphics[width=1.4in]{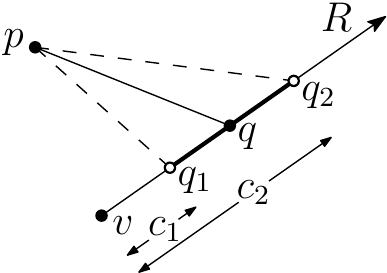}
	\vspace{-12pt} 
\end{wrapfigure}
\noindent{\em Proof. }Let $q_1$ and $q_2$ be the two points on $R$ such that $|vq_1|=c_1$ and $|vq_2|=c_2$. Let $q$ be any point on $R$ with $|pq|>|vq|\geqslant c_1$. We consider two cases where $|vq|\leqslant c_2$ or $|vq|\geqslant c_2$. 
	
First assume that $|vq|\leqslant c_2$, that is, $q$ lies on the segment $q_1q_2$ as in the figure to the right. Then $\max\{|vq|,c_2\}=c_2$, and thus we want the largest value of $|pq|/c_2$. Since the largest distance between a point and a segment is achieved at segment endpoints, the largest value of $|pq|$ is achieved when $q$ is an endpoint of $q_1q_2$. Therefore, the largest value of $|pq|/c$ is achieved when $q=q_1$ or $q=q_2$, that is, when $|vq|=c_1$ or $|vq|=c_2$.
	
Now assume that $|vq|\geqslant c_2$, i.e., $q$ does not lie on $vq_2$. Then $\max\{|vq|,c_2\}=|vq|$, and thus we want the largest value of $|pq|/|vq|$. Since $|pq|>|vq|$, if we decrease $|vq|$ (by moving $q$ towards $q_2$) then $|pq|/|vq|$ increases. Thus, the largest value of $|pq|/|vq|$ is achieved when $|vq|$ is as small as possible, i.e., when $|vq|=c_2$.
\qed
\vspace{8pt}

Now we have adequate tools for our proof of Theorem~\ref{two-angle-thr}. 
	Without loss of generality assume that $\alpha$ is the convex angle at $u$, $\gamma$ is the convex angle at $v$, and $\alpha\leqslant \gamma$. Thus, $60^\circ \leqslant\alpha\leqslant 105^\circ$. After a suitable rotation and/or reflection assume that $uv$ is horizontal, $u$ is to the left of $v$, and both $p$ and $q$ lie above the line through $uv$; see Figure~\ref{Angelini-fig}(b). 
	
	We want to prove that $\sqrt{3}$ is an upper bound on the ratio $|pq|/\max\{|pu|,|uv|,|vq|\}$. 
	We assume that $|pq|> \max\{|pu|,|uv|,|vq|\}$ because otherwise the claim is trivial.	
	To simplify the proof we apply a sequence of geometric transformations that might increase the ratio, but wont decrease it.
	It is implied by Lemma~\ref{Abu-Affash-lemma} that $q$ is outside the triangle $\bigtriangleup puv$. 
	This and the fact that MST is non-crossing, imply that $q$ is to the right side of the ray emanating from $v$ and passing through $p$. Thus, if we rotate $q$ clockwise around $v$ while maintaining the distance $|vq|$, then the angle $\angle pvq$ increases and so does the length $|pq|$; this would increase the objective ratio. Therefore, without loss of generality, we assume that $\gamma$ is as large as possible, i.e., $\gamma=210^\circ-\alpha$.
	
	Since $|pq|>|vq|\geqslant 0$, by Lemma~\ref{moving-point-lemma} we can assume that $|vq| = 0$ or $|vq|=\max\{|pu|,|uv|\}$, where $0$ and $\max\{|pu|,|uv|\}$ play roles of the constants $c_1$ and $c_2$ in this lemma.
	First assume that $|vq|=0$, i.e., $q=v$. Lemma~\ref{length-lemma} implies that $|pq|\leqslant 2\sin(\alpha/2)\cdot \max\{|pu|,|uv|\}$. Since $\sin(\alpha/2)$ is increasing on the interval $\alpha\in[60^\circ,105^\circ]$, its largest value is achieved at $\alpha=105^\circ$. Thus $2\sin(\alpha/2)\leqslant 2\sin(52.5^\circ) < \sqrt{3}$, and our claim follows.

	\begin{figure}[htb]
		\centering
		\setlength{\tabcolsep}{0in}
		$\begin{tabular}{ccc}
		\multicolumn{1}{m{.33\columnwidth}}{\centering\includegraphics[width=.33\columnwidth]{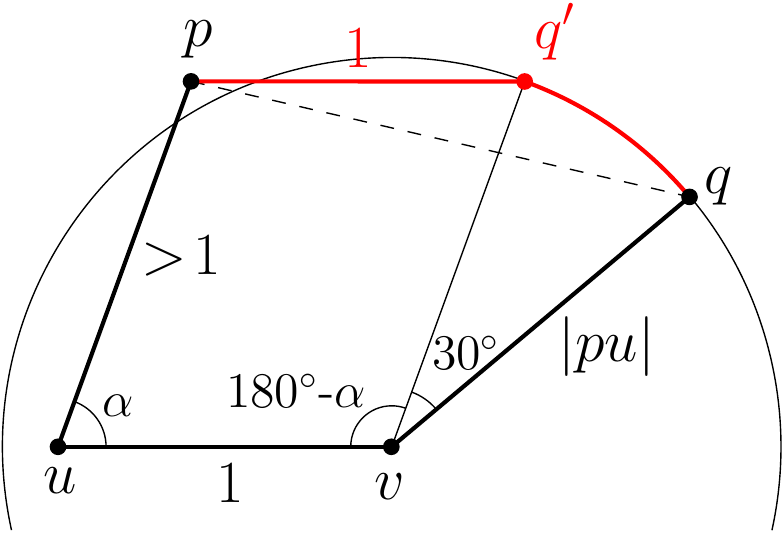}}
		&\multicolumn{1}{m{.33\columnwidth}}{\centering\includegraphics[width=.26\columnwidth]{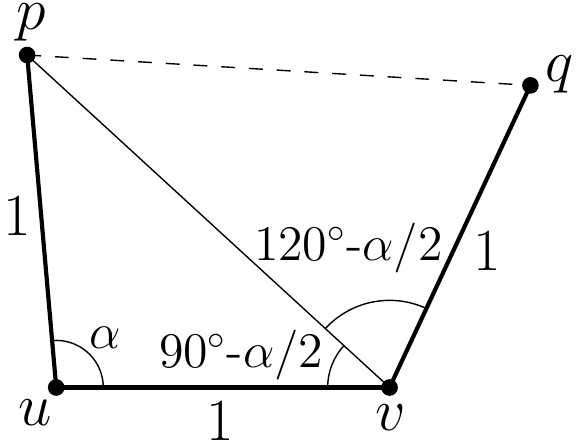}}
		&\multicolumn{1}{m{.33\columnwidth}}{\centering\includegraphics[width=.32\columnwidth]{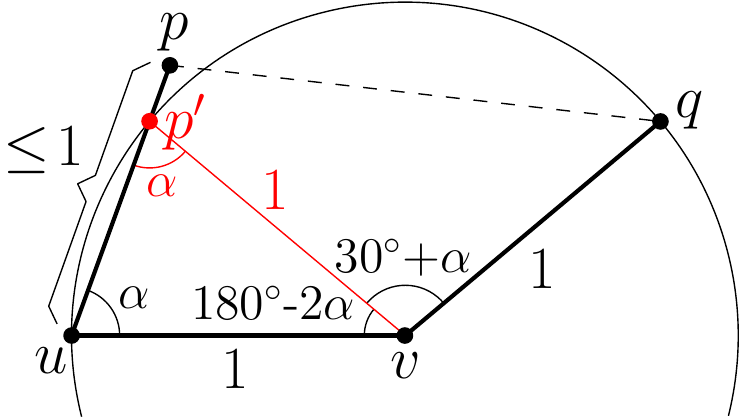}}\\
		(a) \text{$|pu| >1$ and $|vq|=|pu|$} & (b) \text{$|pu|= 1$ and $\alpha\geqslant 90^\circ$}&(c) \text{$|pu|\leqslant 1$ and $\alpha< 90^\circ$} 
		\end{tabular}$
		\caption{Illustration of the proof of Theorem~\ref{two-angle-thr}.}
		\label{x-larger-fig}
	\end{figure}
	
	Now assume that $|vq|=\max\{|pu|,|uv|\}$. After a suitable rescaling assume that $|uv|=1$. We consider two cases: $|pu|>1$ and $|pu|\leqslant 1$.
	
	\begin{itemize}
		\item $|pu|>1$. In this case $|vq|=\max\{|pu|,1\}=|pu|$. Consider the circle of radius $|pu|$ that is centered at $v$, as in Figure~\ref{x-larger-fig}(a); this circle contains $q$. Let $q'$ be the point obtaining by moving $q$ on this circle counterclockwise  for $30$ degrees. Then $|vq'|=|pu|$ and $vq'$ is parallel to $pu$ (because $\alpha+ \angle uvq'=180^\circ$). Thus $p,u,v,q'$ are vertices of a parallelogram, and hence $|pq'|=1$. The length of $pq$ is at most $|pq'|$ plus the length of the arc $\widehat{q'q}$ which is $|pu|\cdot\pi/6$. In other words,
		$|pq|\leqslant 1+ |pu|\cdot\pi/6< |pu|+|pu|\cdot\pi/6< \sqrt{3}\cdot |pu|$. Thus our claim follows.
		
		\item $|pu|\leqslant 1$. Then $|vq|=\max\{|pu|,1\}=1$. We consider two subcases where $\alpha\geqslant 90^\circ$ or $\alpha< 90^\circ$.
		
		\begin{itemize}
			\item $\alpha \geqslant 90^\circ$. 
			Since $|pq|>|pu|\geqslant 0$, by Lemma~\ref{moving-point-lemma} we can assume that $|pu|=0$ or $|pu| = \max\{|uv|,|vq|\}=1$ (to apply the lemma notice that $u$, $0$, $\max\{|uv|,|vq|\}$ play roles of $v$, $c_1$, $c_2$, respectively, and the roles of $p$ and $q$ are swapped).
			First assume that $|pu|=0$, i.e., $p=u$. Since $\bigtriangleup uvq$ is isosceles with vertex angle $210^\circ -\alpha$, we get $|pq|= 2\sin(105^\circ -\alpha/2)$. This function is decreasing on the interval $\alpha\in[90^\circ, 105^\circ]$, and its largest value $\sqrt{3}$ is attained at $\alpha=90^\circ$. Thus our claim follows. Now assume that $|pu|=1$, as in Figure~\ref{x-larger-fig}(b). Then $|pv|=2\sin (\alpha/2)$ because $\bigtriangleup puv$ is isosceles. By the law of cosines (on $\bigtriangleup pvq$) we have
			\[
			|pq|=\sqrt{1+4\sin^2(\alpha/2)-4\sin(\alpha/2)\cos(120^\circ-\alpha/2)},
			\]
			which is increasing on the interval $\alpha\in[90^\circ, 105^\circ]$ and its largest value $1+\sqrt{2-\sqrt{3}}\approx 1.518$ is attained at $\alpha=105^\circ$. 
			
			\item $\alpha< 90^\circ$. Notice that the points $p$, $u$, and $v$ have not been moved by above transformations, and  thus they are still vertices of the MST. By minimality of the MST, $p$ does not lie in the interior of the circle with radius 1 that is centered at $v$; see Figure~\ref{x-larger-fig}(c). This and our assumption $\alpha< 90^\circ$ imply that $pu$ intersects this circle at a point other than $u$; let $p'$ denote this intersection point. Thus, we have $|pu|\geqslant |p'u|= 2\sin(90^\circ-\alpha)$; the equality holds because $\bigtriangleup p'uv$ is isosceles with vertex angle $180^\circ-2\alpha$.

			Since $|pq|>|pu|\geqslant 2\sin(90^\circ-\alpha)$, by Lemma~\ref{moving-point-lemma} we can assume that $|pu|=2\sin(90^\circ-\alpha)$ or $|pu| = \max\{|uv|,|vq|\}=1$ (notice that $2\sin(90^\circ-\alpha)$ and $\max\{|uv|,|vq|\}$ play roles of $c_1$ and $c_2$, respectively). If $|pu|=2\sin(90^\circ-\alpha)$, i.e., $p=p'$,  then $\bigtriangleup pvq$ is isosceles with vertex angle $30^\circ+\alpha$; see Figure~\ref{x-larger-fig}(c). In this case we get $|pq|=2\sin(15^\circ+\alpha/2)$, which is at most $\sqrt{3}$ on the interval $\alpha\in [60^\circ,90^\circ)$. Assume that $|pu|=1$. Then $|pv|=2\sin(\alpha/2)$ because $\bigtriangleup puv$ is isosceles. By the law of cosines (on $\bigtriangleup pvq$) we have 
			\[|pq|=\sqrt{1+4\sin^2(\alpha/2)-4\sin(\alpha/2)\cos(120^\circ-\alpha/2)},
			\]
			which is increasing on $\alpha\in[60^\circ, 90^\circ)$ with largest value at most $\sqrt{4-\sqrt{3}}\approx 1.506$. \qedhere
		\end{itemize}
	\end{itemize}

\section{Conclusions}
A natural open problem is to improve the upper bounds on $\beta_2$, $\beta_3$, $\beta_4$ further by designing better algorithms. Specifically, it is natural to investigate whether the geometry
of the Euclidean plane (besides the triangle inequality) can be exploited to develop an approximation algorithm for the bottleneck degree-2 spanning tree problem (i.e., the bottleneck Hamiltonian path problem) with factor less than $3$. We note the existence of a 2-approximation algorithm \cite{Parker1984} for the bottleneck Hamiltonian cycle problem in general metric spaces; this algorithm is based on the Hamiltonicity of the square of every biconnected graph~\cite{Fleischner1974}.

The study of worst-case ratios in higher dimensions is more vital as the maximum degree of an MST and a BST can be much larger. Zbarsky \cite{Zbarsky2014} showed that $1.447\leqslant \tau_3\leqslant 1.559$ in the $d$-dimensional Euclidean space for any $d\geqslant 2$. Andersen and Ras\cite{Andersen2108b} studied the bottleneck version of the problem in dimension 3.  
 
\paragraph{Acknowledgement.}
I thank Jean-Lou De Carufel for helpful suggestions on simplifying the proof of Lemma~\ref{moving-point-lemma}.

\bibliographystyle{abbrv}
\bibliography{Bounded-Degree-BST}
\end{document}